%% file: main.tex
\documentclass[conference,letterpaper]{IEEEtran}

\usepackage[utf8]{inputenc} 
\usepackage[T1]{fontenc}
\usepackage{url}
\usepackage{ifthen}
\usepackage{cite}
\usepackage[cmex10]{amsmath} 

\usepackage{fancyhdr} 
\usepackage{bm}
\usepackage{amsfonts}
\usepackage{amssymb}
\usepackage{mathtools}
\mathtoolsset{showonlyrefs=true}
\usepackage{cite}
\usepackage{multirow}
\newtheorem{definition}{Definition}
\newtheorem{theorem}{Theorem}
\newtheorem{assumption}{Assumption}
\newtheorem{proof}{Proof}
\newtheorem{corollary}{Corollary}
\newtheorem{lemma}{Lemma}

\newcommand{\qed}{\hfill$\Box$}

\fancypagestyle{firstpage}{
  \fancyhf{}
  \fancyfoot[L]{\tiny $\copyright$ 2024 IEEE. Personal use of this material is permitted. Permission from IEEE must be obtained for all other uses, in any current or future media, including reprinting/republishing this material for advertising or promotional purposes, creating new collective works, for resale or redistribution to servers or lists, or reuse of any copyrighted component of this work in other works.}
}

\begin{document}
\title{Unbiased Estimating Equation on Inverse Divergence and Its Conditions} 

\author{%
  \IEEEauthorblockN{Masahiro Kobayashi}
  \IEEEauthorblockA{Information and Media Center\\
                    Toyohashi University of Technology\\
                    Toyohashi, Aichi, Japan\\
                    Email: kobayashi@imc.tut.ac.jp}
  \and
  \IEEEauthorblockN{Kazuho Watanabe}
  \IEEEauthorblockA{Department of Computer Science and Engineering\\ 
                    Toyohashi University of Technology\\
                    Toyohashi, Aichi, Japan\\
                    Email: wkazuho@cs.tut.ac.jp}
}

\maketitle

\begin{abstract}
This paper focuses on the Bregman divergence defined by the reciprocal function, called the inverse divergence.
For the loss function defined by the monotonically increasing function $f$ and inverse divergence, the conditions for the statistical model and function $f$ under which the estimating equation is unbiased are clarified. 
Specifically, we characterize two types of statistical models, an inverse Gaussian type and a mixture of generalized inverse Gaussian type distributions, to show that the conditions for the function $f$ are different for each model. 
We also define Bregman divergence as a linear sum over the dimensions of the inverse divergence and extend the results to the multi-dimensional case.
\end{abstract}

\thispagestyle{firstpage} 

\input{1_Introduction}
\input{2_f-seprabale}
\input{3_Inverse_divergence}
\input{4_Multi-dimensional_extension}
\input{5_Conclusion}

\small
\section*{Acknowledgment}
Except the abstract, this paper was translated from Japanese into English by OpenAI ChatGPT-4.
We would like to thank Editage and Enago for English language editing.
This work was supported in part by JSPS KAKENHI Grant numbers JP23K16849 and JP19K11825.
\normalsize

\enlargethispage{-1.2cm} 


\bibliographystyle{IEEEtran}
\bibliography{Refj}

\clearpage

\input{Appendix}

\end{document}

%% file: 1_Introduction.tex
\section{Introduction}
The maximum likelihood estimation (MLE) is a standard method in parametric estimation although it is vulnerable to outliers. 
In robust statistics, the methods studied to overcome the adverse effects of outliers \cite{robust, huberbook} include M-estimation, which is a well-known robust estimation method. 
In M-estimation, the assumed model in the MLE is changed to another with heavy tails. 
The well-known minimum divergence method changes the Kullback--Leibler divergence corresponding to the MLE to a more robust divergence for estimation \cite{basubook, pardobook}. 
The M-estimation and minimum divergence methods obtain estimators as solutions to estimating equations, which are mainly of two types: non-normalized and normalized.
Bregman divergence and its special cases \cite{beta-div, bexp-div, Roy2019, Singh2021, u-boost} correspond to the non-normalized estimating equation, whereas $\gamma$-divergence \cite{gamma-div} corresponds to the normalized estimating equation. 
Notably, the normalized estimating equation has the distinctive feature of bringing the latent bias close to zero, even in cases heavily contaminated by outliers \cite{fujisawa2013}. 
However, the analytically intractable integrals involved in these estimating equations limit the choice of models and weight functions that can be used. 
Recent studies on the minimization of divergence have adopted the stochastic optimization framework to avoid these intractable integrals \cite{okuno2023minimizing, pscd2021}.

The $f$-separable distortion measure \cite{f-separable} was proposed as an extension of the average and maximum distortions that are commonly used in information theory. 
Recently, we proposed a parameter estimation method that minimizes the $f$-separable distortion using Bregman divergence as a base distortion measure \cite{mypaper3}.
This method is a type of M-estimation method.
The property of the estimator is determined by the shape of the monotonically increasing function $f$, exhibiting robustness against outliers when $f$ is concave. 
However, extending the loss function using function $f$ does not always guarantee an unbiased estimating equation, which is a necessary condition for the consistency of the estimator. 
To satisfy the unbiased estimating equation, a bias correction term involving an analytically intractable integral is required. 
Consequently, to avoid this problem, combinations of the statistical model, Bregman divergence, and function $f$ that satisfy the unbiased estimating equation without a bias correction term have been investigated \cite{mypaper5}.
It has been suggested that the necessity of a bias correction term depends on the type of Bregman divergence utilized in the estimation \cite{mypaper5}. 
When a bias correction term is not required, the Bregman divergence describes the corresponding statistical model used for estimation, and the available functions $f$ are characterized by the boundedness of a specific simple integral. 
In rare cases, the bias correction term disappears; however, the combinations generating this condition are unknown.

This paper focuses on the Bregman divergence defined by the reciprocal function, called the inverse divergence.
We clarify combinations of the statistical model and function $f$ that eliminate the bias correction term when using inverse divergence for estimation.
Furthermore, we extend the result to a multi-dimensional  case by expressing the inverse divergence as a linear sum in multiple dimensions.

%% file: 2_f-seprabale.tex
\renewcommand{\arraystretch}{1.5}
\begin{table*}[h]
    \centering
    \caption{Combination of Bregman divergence, statistical model, and function $f$ when the bias correction term vanishes}
    \small
    \begin{tabular}{c|c|c|c|c}
         Divergence & Support & Model & PDF & Bounded condition \\
         \hline\hline
         \begin{tabular}{c}
            Mahalanobis \\
            $d_{\rm M}^{\bm{A}}(\bm{x},\bm{\theta})=(\bm{x}-\bm{\theta})^{\rm T}\bm{A}(\bm{x}-\bm{\theta})$ 
         \end{tabular}
          & $\mathbb{R}^d$ & Elliptical\cite{elliptical_book} & $\frac{|\bm{A}|^\frac{1}{2}}{C}g\left(d_{\rm M}^{\bm{A}}(\bm{x},\bm{\theta})\right)$ & $\int_0^\infty g(t)f'(t)t^\frac{d-1}{2}dt$ \\ \hline
         \begin{tabular}{c}
            1D Bregman$^1$\\
            $d_\phi(x,\theta)$ 
         \end{tabular}
        & $(a,b)$ & Continuous Bregman\cite{mypaper5} & $\frac{1}{C_\phi(\theta)}\frac{\phi'(x)-\phi'(\theta)}{x-\theta}g\left(d_\phi(x,\theta)\right)$ & $\int_0^\zeta g(t)f'(t)dt$ \\ \hline
        \begin{tabular}{c}
            Inverse  \\
            \eqref{eq:ig_div} 
        \end{tabular}
        & $\mathbb{R}_{++}$ & \begin{tabular}{c}
             Inverse Gaussian type\\
             (IGT) \cite{igt}
        \end{tabular} & \eqref{eq:igt} & \eqref{eq:bregman_ig_cond} [This paper]\\
        \begin{tabular}{c}
            Multivariate inverse  \\
            \eqref{eq:multivariate_ig_div} 
        \end{tabular}
        & $\mathbb{R}_{++}^d$ & \begin{tabular}{c}
             Multivariate IGT\\
             (MIGT) [This paper]
        \end{tabular} & \eqref{eq:multi_igt} & \eqref{eq:multi_bregman_ig_cond} [This paper]\\
        \multicolumn{5}{l}{$^{\mathrm{1}}$The one-dimensional Bregman divergence that satisfies the following condition:} \\
        \multicolumn{5}{l}{$\forall \theta \in \Theta, \; \lim_{x\to a} d_\phi(x,\theta) = \lim_{x\to b} d_\phi(x, \theta) = \zeta \in \mathbb{R}_{++} \cup \{\infty\}$, where $\mathbb{R}_{++}$ is the set of positive real numbers.}
    \end{tabular}
    \label{tab:vanished_list}
    \vspace{-10.5pt}
\end{table*}
\renewcommand{\arraystretch}{1}
\renewcommand{\arraystretch}{1.5}
\begin{table*}[h]
    \centering
    \normalsize
    \caption{Special case of continuous Bregman distribution and corresponding Bregman divergence}
    \small
    \begin{tabular}{c|c|c|c|c|c}
         $\phi(x)$& Divergence & Support & Generator & Model & PDF \\
         \hline\hline
         \multirow{2}{*}{\large$\frac{x^2}{\sigma^2}$}\normalsize&\multirow{2}{*}{\begin{tabular}{c}
            Squared \\
            $\frac{(x-\theta)^2}{\sigma^2}$\normalsize
         \end{tabular} }
         & \multirow{2}{*}{$\mathbb{R}$} & $g(t)$ & 1D elliptical\cite{elliptical_book}& $\frac{1}{C\sigma}g(\frac{(x-\theta)^2}{\sigma^2})$   \\
         & & & $\exp(-\frac{t}{2})$ & Gaussian & $\frac{1}{\sqrt{2\pi\sigma^2}}\exp(-\frac{(x-\theta)^2}{2\sigma^2})$  \\ \hline
        \multirow{2}{*}{$-k\log x$}&\multirow{2}{*}{\begin{tabular}{c}
            Itakura--Saito (IS) \\
            $d_{\rm IS}^k(x,\theta) = k(\frac{x}{\theta}-\log\frac{x}{\theta}-1)$ 
         \end{tabular} }
         & \multirow{2}{*}{$\mathbb{R}_{++}$} & $g(t)$ & IS\cite{mypaper5}& $\frac{1}{C(k)}\frac{1}{x}g(d_{\rm IS}^k(x, \theta))$   \\ 
         & & & $\exp(-t)$ & Gamma & $\left(\frac{k}{\theta}\right)^k \frac{x^{k-1}}{\Gamma(k)}\exp(-\frac{k}{\theta}x)$   \\ \hline
         \multirow{2}{*}{\large$\frac{\lambda}{x}$}\normalsize&\multirow{2}{*}{\begin{tabular}{c}
            Inverse \\
            \eqref{eq:ig_div} 
         \end{tabular} }
         & \multirow{2}{*}{$\mathbb{R}_{++}$} & $g(t)$ & GIGT mixture [This paper]& \eqref{eq:mixture_cbd}  \\ 
         & & & $\exp(-\frac{t}{2})$ & GIG mixture [This paper] & \eqref{eq:gig_mixture}\\
         \multicolumn{6}{l}{Abbreviations: GIGT, Generalized inverse Gaussian type; GIG, Generalized inverse Gaussian.}
    \end{tabular}
    \label{tab:cbd_list}
    \vspace{-10.5pt}
\end{table*}
\renewcommand{\arraystretch}{1}

\section{$f$-separable Bregman distortion measures}
\subsection{Problem setting and loss function}
In estimating the parameter $\bm{\theta}\in\bm{\Theta}\subseteq\mathbb{R}^d$ of a statistical model $p(\bm{x}|\bm{\theta})$ based on the given data $\bm{x}^n=\{\bm{x}_1,\cdots,\bm{x}_n\}, \; \bm{x}_i=(x_i^{(1)},\cdots,x_i^{(d)})^{\rm T}\in\bm{\chi}\subseteq\mathbb{R}^d$, we assume that the true distribution can be realized by $p(\bm{x}|\bm{\theta}^*)$. 
When the expected value exists for the statistical model, we assume that $\bm{\theta}=\mathbb{E}[\bm{X}]=\int \bm{x}p(\bm{x}|\bm{\theta})d\bm{x}$.
The loss function is defined as
\begin{align}
    L(\bm{\theta}) = \frac{1}{n}\sum_{i=1}^n f(d_\phi(\bm{x}_i, \bm{\theta})), \label{eq:loss}
\end{align}
using differentiable and continuous monotonically increasing function $f:\mathbb{R}_+\to\mathbb{R}$ and the Bregman divergence $d_\phi(\bm{x},\bm{\theta}):\bm{\chi}\times\bm{\Theta}\to\mathbb{R}_+$, where $\mathbb{R}_+$ is the set of nonnegative real numbers \cite{mypaper3, mypaper5, f-separable}.
The Bregman divergence is defined as
\begin{align*}
    d_\phi(\bm{x}, \bm{\theta}) = \phi(\bm{x}) - \phi(\bm{\theta}) - \langle \bm{\nabla}\phi(\bm{\theta}), \bm{x}-\bm{\theta} \rangle,
\end{align*}
where $\phi:\bm{\chi}\to\mathbb{R}$  is a differentiable strictly convex function with gradient vector $\bm{\nabla}\phi$ and $\langle \cdot, \cdot \rangle$ denotes the inner product. 
The property of the estimator depends on the shape of the function $f$. 
When the function $f$ is concave, the estimator is robust against outliers. 
Furthermore, when the function $f$ is linear, the estimation problem reduces to the MLE for the expected value parameter in a regular exponential family:
\begin{align}   
    p(\bm{x}|\bm{\theta}) = r_\phi(\bm{x})\exp(-d_\phi(\bm{x}, \bm{\theta})), \label{eq:reg_exp}
\end{align}
where the strictly convex function $\phi$ uniquely determines $r_\phi(\bm{x})$ \cite{clustering-bregman}.

\subsection{Conditions for unbiased estimating equation}
The estimator $\hat{\bm{\theta}}$ is the solution to the stationary point, obtained by differentiating the loss function \eqref{eq:loss} with respect to the parameter $\bm{\theta}$ and setting it to $\bm{0}$, in solving the estimating equation:
\begin{align}
    \frac{1}{n}\sum_{i=1}^n f'(d_\phi(\bm{x}_i, \bm{\theta}))(\bm{x}_i - \bm{\theta}) = \bm{0}, \label{eq:est_eq}
\end{align}
where $f'$ is the derivative of $f$. 
However, this estimating equation is generally biased. 
This is because the left-hand side of the estimating equation \eqref{eq:est_eq} does not necessarily converge asymptotically to $\bm{0}$. 
In the limiting value, bias should be pre-corrected to satisfy the unbiased estimating equation. 
Two main types of unbiased estimating equations are known: non-normalized \eqref{eq:fueq} and normalized \eqref{eq:fneq}. 
The latter is particularly noted for its ability to reduce latent biases to $\bm{0}$, even in situations with a high proportion of outliers:
\small
\begin{align}
    \frac{1}{n}\sum_{i=1}^n f'(d_\phi(\bm{x}_i, \bm{\theta}))(\bm{x}_i - \bm{\theta}) = \mathbb{E}_{p(\bm{x}|\bm{\theta})}\left[f'(d_\phi(\bm{X}, \bm{\theta}))(\bm{X}-\bm{\theta})\right], \label{eq:fueq}\\
    \frac{\sum_{i=1}^n f'(d_\phi(\bm{x}_i, \bm{\theta}))(\bm{x}_i-\bm{\theta})}{\sum_{j=1}^n f'(d_\phi(\bm{x}_j, \bm{\theta}))} = \frac{\mathbb{E}_{p(\bm{x}|\bm{\theta})}\left[f'(d_\phi(\bm{X}, \bm{\theta}))(\bm{X}-\bm{\theta})\right]}{\mathbb{E}_{p(\bm{x}|\bm{\theta})}\left[f'(d_\phi(\bm{X}, \bm{\theta}))\right]}, \label{eq:fneq}
\end{align}
\normalsize
\cite{fujisawa2013, mypaper5}.
The bias correction term on the right hand side of each estimating equation is generally difficult to integrate analytically. 
However, if the following condition holds, the bias correction terms on the right sides of both the non-normalized \eqref{eq:fueq} and normalized \eqref{eq:fneq} estimating equations vanish, resulting in the same problem expressed by the estimating equation \eqref{eq:est_eq}:
\begin{align}
    \mathbb{E}_{p(\bm{x}|\bm{\theta})}\left[f'(d_\phi(\bm{X}, \bm{\theta}))(\bm{X}-\bm{\theta})\right] = \bm{0}. \label{eq:bias0}
\end{align}
In other words, if \eqref{eq:bias0} is satisfied, \eqref{eq:est_eq} is the unbiased estimating equation not requiring a bias correction term.
Additionally, as it can be interpreted as the normalized estimating equation \eqref{eq:fneq}, the latent bias can be approximately reduced to $\bm{0}$, even in situations with a high proportion of outliers.

\subsection{Existing results}
Table \ref{tab:vanished_list} shows the combinations of Bregman divergence, statistical model, and function $f$ for which the bias correction term vanishes, and \eqref{eq:bias0} is satisfied. 
Previous research has clarified the conditions for the statistical model, function $f$ corresponding to the Mahalanobis distance, and one-dimensional Bregman divergence\footnotemark[1] \cite{mypaper5}. 
The Mahalanobis distance and one-dimensional Bregman divergence correspond to elliptical and continuous Bregman distributions, respectively.
Elliptical distributions are a family of distributions defined by the generating function and the Mahalanobis distance \cite{elliptical_book}. 
They include the well-known Gaussian, Laplace, and  $t$- distributions as special cases.
Whereas, continuous Bregman distributions are a distribution family defined by the generating function $g$ and a strictly convex function $\phi$. 
The special cases of this family are the one-dimensional elliptical distribution and the Itakura--Saito (IS) distribution, which generalizes the gamma distribution (Table \ref{tab:cbd_list}).

The loss function \eqref{eq:loss} is an extension of the negative log-likelihood function under the regular exponential family \eqref{eq:reg_exp}, enhanced by the function $f$. 
Therefore, it is preferable that statistical models that satisfy the unbiased estimating equation \eqref{eq:bias0} correspond to the regular exponential family. 
However, apart from the Gaussian and gamma distributions, no special cases of the regular exponential family have been identified as satisfying the unbiased estimating equation \eqref{eq:bias0}. 
The following sections clarify that for estimations using inverse divergence, the statistical model satisfying the unbiased estimating equation \eqref{eq:bias0} corresponds to the inverse Gaussian distribution (and its generalized distribution family), which is a special case of the regular exponential family.

%% file: 3_Inverse_divergence.tex
\section{Models and conditions under inverse divergence}
This section discusses the Bregman divergence specified by the strictly convex function $\phi(x)=\lambda/x,\; (x>0, \lambda>0)$ which is called the inverse divergence \cite[p.102]{nmf_book}, where $\lambda$ is a nuisance parameter. 
The inverse divergence, defined as
\begin{align}
    d_{\rm Inv}^\lambda(x, \theta) \triangleq \frac{\lambda(x-\theta)^2}{\theta^2x}, \label{eq:ig_div}
\end{align}
is known to correspond to the inverse Gaussian distribution \cite{clustering-bregman}, \cite[Chapter 2]{Wuthrich2023} and is also a particular case of the $\beta$-divergence \cite{Cichocki2010}. 
In the following subsections, we discuss the two types of statistical models satisfying the unbiased estimating equation \eqref{eq:bias0} when using inverse divergence for estimation and show that the conditions for the function $f$ differ for each model.

\subsection{Inverse Gaussian type (IGT) distribution}
\begin{definition}[IGT distribution \cite{igt}]
    For $x\in\mathbb{R}_{++}$, the parameters $\theta\in\Theta=\mathbb{R}_{++}$ and $\lambda\in\mathbb{R}_{++}$, and a nonnegative generating function $g:\mathbb{R}_+\to\mathbb{R}_+$, the IGT distribution is defined as follows:
    \begin{align}
        p(x|\theta, \lambda) &= \frac{1}{C_{\rm IGT}}\sqrt{\frac{\lambda}{x^3}}g(d_{\rm Inv}^\lambda(x, \theta)), \label{eq:igt}\\
        C_{\rm IGT} &= \int_0^\infty \frac{1}{\sqrt{t}}g(t)dt, \label{eq:const_igt}
    \end{align}
    if the normalization constant $C_{\rm IGT}$ exists.
\end{definition}
Here, $\mathbb{R}_{++}=\mathbb{R}_{+} \setminus \{0\}$.
Note that the normalization constant $C_{\rm IGT}$ is the same as that for the one-dimensional elliptical distribution \cite{elliptical_book, igt}.
When the generating function is $g(t)=\exp(- t/2)$, the IGT distribution \eqref{eq:igt} reduces to the inverse Gaussian distribution \cite{chhikara1988inverse, seshadri1994inverse, seshadri1999inverse} as follows:
\begin{align*}
    p(x|\theta, \lambda) = \sqrt{\frac{\lambda}{2\pi x^3}}\exp\left(-\frac{\lambda (x-\theta)^2}{2\theta^2x}\right).
\end{align*}
The expected value of the IGT distribution, $\mathbb{E}[X]=\theta$, is independent of the generating function $g$ \cite{igt}. 
This fact can also be derived from the following Corollary \ref{corollary:igt_expect}, which is derived below.
\begin{assumption}\label{assp:igt}
    There exists an IGT distribution \eqref{eq:igt} corresponding to the nonnegative generating function $g$, i.e., $C_{\rm IGT}<\infty$.
\end{assumption}
\begin{theorem} \label{theorem:igt}
    Under Assumption \ref{assp:igt}, the estimating equation without a bias correction term equivalently, \eqref{eq:bias0} holds if and only if
    \begin{align}
        \int_0^\infty g(t)f'(t)\frac{1}{\sqrt{t+1}}dt<\infty \label{eq:bregman_ig_cond}
    \end{align}
    holds for the combination of the function $f$ and the statistical model \eqref{eq:igt}.
\end{theorem}
\begin{proof}
    Substituting the inverse divergence \eqref{eq:ig_div} and IGT distribution \eqref{eq:igt} into the left-hand side of \eqref{eq:bias0}, we have
    \begin{align}
    &\mathbb{E}_{p(x|\theta, \lambda)}\left[f'(d_{\rm Inv}^\lambda(X, \theta))(X-\theta)\right] \nonumber \\
    =&\int_0^\infty \frac{1}{C_{\rm IGT}}\sqrt{\frac{\lambda}{x^3}}g(d_{\rm Inv}^\lambda(x, \theta))f'(d_{\rm Inv}^\lambda(x, \theta))(x-\theta)dx \nonumber \\
    \propto& \int_0^\theta \frac{1}{\sqrt{x^3}}g(d_{\rm Inv}^\lambda(x, \theta))f'(d_{\rm Inv}^\lambda(x, \theta))(x-\theta)dx \nonumber \\
    &+ \int_\theta^\infty \frac{1}{\sqrt{x^3}}g(d_{\rm Inv}^\lambda(x, \theta))f'(d_{\rm Inv}^\lambda(x, \theta))(x-\theta)dx \label{eq:proof_ig1}\\
    =&\frac{\theta}{\sqrt{\lambda}} \int_0^\infty g(t)f'(t) \left[\frac{1}{\sqrt{t+4\frac{\lambda}{\theta}}}-\frac{1}{\sqrt{t+4\frac{\lambda}{\theta}}}\right]dt = 0. \nonumber
\end{align}
    We used integration by substitution, $t=d_{\rm Inv}^\lambda(x, \theta)$. 
    The details of the substitution integration from \eqref{eq:proof_ig1} to the next line are provided in Appendix \ref{sec:appendix2}.
    Therefore, if the following integral exists for any $a>0$, the unbiased estimating equation \eqref{eq:bias0} holds without a bias correction term:
    \begin{align}
        I(a) \triangleq \int_0^\infty g(t)f'(t)\frac{1}{\sqrt{t + a}}dt<\infty. \label{eq:I-theta}
    \end{align}
    Conversely, the above discussion also shows that $\mathbb{E}_{p(x|\theta, \lambda)}\left[\left|f'(d_{\rm Inv}^\lambda(X,\theta))(X-\theta)\right|\right]\propto 2\theta\sqrt{\lambda}^{-1}I(4\lambda\theta^{-1})$.
    In other words, $I(a)<\infty$ is also a necessary condition.

    However, since this integral includes the parameter $a=4\lambda/\theta>0$, rewriting it in a form that does not depend on $a$ is desirable. 
    The integrand of \eqref{eq:I-theta} is a strictly monotonically decreasing and continuous function with respect to $a$. 
    Therefore, the $I(a)$ obtained through integration is also a strictly monotonically decreasing and continuous function with respect to $a$.
    We assume that $I(a^\dagger)<\infty$ for $\exists a^\dagger \in (0,\infty)$.
    The strict monotonic decrease of $I(a)$ ensures that $\forall \varepsilon>0, I(a^\dagger +\varepsilon)<\infty$ holds. 
    Furthermore, owing to the continuity of $I(a)$, $\forall \varepsilon\in(0,a^\dagger), I(a^\dagger - \varepsilon)<\infty$ holds. 
    Therefore, if $I(a)$ is bounded at some point $a^\dagger\in(0,\infty)$, it is bounded for any $a\in(0,\infty)$.
    Thus, we set $a=1$.
    Based on the above discussion, \eqref{eq:bregman_ig_cond} is a necessary and sufficient condition for the unbiased estimating equation to hold without the bias correction term.
    \qed
\end{proof}

The condition $\int_0^\infty g(t)\sqrt{t+1}^{-1}dt<\infty$ for the existence of the expected value of the IGT distribution is obtained from Theorem \ref{theorem:igt} by substituting $f'(t)=1$.
As this condition is the lower bound of the normalization constant \eqref{eq:const_igt} of the IGT distribution, the following corollary is obtained.
\begin{corollary}\label{corollary:igt_expect}
    The expected value of the IGT distribution \eqref{eq:igt} always exists, independent of the generating function $g$, and satisfies $\mathbb{E}[X]=\theta$.
\end{corollary}

In the above discussion, we showed that for estimation using the inverse divergence \eqref{eq:ig_div}, the statistical model satisfying the unbiased estimating equation \eqref{eq:bias0} corresponds to the IGT distribution.
The inverse Gaussian distribution, a special case of the IGT distribution, is also a special case of the regular exponential family of distributions.
On the other hand, continuous Bregman distributions correspond to the one-dimensional Bregman divergence. 
Thus, the inverse divergence also corresponds to continuous Bregman distributions generated by the reciprocal function. 
In the following, we show that generalized inverse Gaussian type (GIGT) mixture distributions, which are special cases of continuous Bregman distributions, can be generated and correspond to the inverse divergence.

\subsection{Generalized IGT (GIGT) mixture distribution}
\begin{definition}[GIGT distribution]
    For $x\in\mathbb{R}_{++}$, the parameters $\theta\in\Theta=\mathbb{R}_{++}$, $\lambda\in\mathbb{R}_{++}$, $\nu\in\mathbb{R}$, and a nonnegative generating function $g:\mathbb{R}_+\to\mathbb{R}_+$, the GIGT distribution is defined as follows:
    \begin{align}
        q(x|\theta, \lambda, \nu) &= \frac{1}{C_{\rm GIGT}(\theta, \lambda, \nu)}x^{\nu-1}g(d_{\rm Inv}^\lambda (x,\theta)), \label{eq:gigt} \\
        C_{\rm GIGT}(\theta, \lambda, \nu) &= \int_0^\infty t^{\nu-1}g(d_{\rm Inv}^\lambda (t,\theta))dt, \nonumber
    \end{align}
    if the normalization constant $C_{\rm GIGT}$ exists.
\end{definition}
When $\nu=-1/2$, the GIGT distribution \eqref{eq:gigt} reduces to the IGT distribution \eqref{eq:igt}.
In continuous Bregman distributions, when the strictly convex function is set as $\phi(x)=\lambda/x, \; (x>0, \lambda>0)$, the corresponding statistical model is defined by a two-component mixture of GIGT distribution \eqref{eq:gigt} as
\begin{align}
    &p(x|\theta, \lambda) = w q(x|\theta, \lambda, 0) + (1-w)q(x|\theta, \lambda, -1), \label{eq:mixture_cbd}\\
    &w = \frac{C_{\rm GIGT}(\theta, \lambda, 0)}{C_{\rm GIGT}(\theta, \lambda, 0)+\theta C_{\rm GIGT}(\theta, \lambda, -1)}.
\end{align}
Specifically, when the generating function is $g(t)=\exp(-t/2)$, the GIGT distribution \eqref{eq:gigt} reduces to the generalized inverse Gaussian (GIG) distribution \eqref{eq:gig} \cite[p.6]{jorgsen1982}, and the GIGT mixture distribution \eqref{eq:mixture_cbd} becomes the GIG mixture distribution \eqref{eq:gig_mixture}:
\begin{align}
    &p(x|\theta, \lambda) = w p(x|\alpha, \theta, 0) + (1-w) p(x|\alpha, \theta, -1), \label{eq:gig_mixture} \\
    &p(x|\alpha, \eta, \nu) = \frac{\eta^{-\nu}x^{\nu-1}}{2K_\nu(\alpha)}\exp\left(-\frac{\alpha}{2}\left(\frac{x}{\eta}+\frac{\eta}{x}\right)\right), \label{eq:gig} \\
    &w = \frac{K_0\left(\alpha\right)}{K_0\left(\alpha\right)+K_{-1}\left(\alpha\right)}, \; \alpha = \frac{\lambda}{\theta}, \nonumber
\end{align}
where $K_\nu(\cdot)$ represents the modified Bessel function of the third kind with index $\lambda$.

The combination of the one-dimensional Bregman divergence and continuous Bregman distribution has been shown to satisfy the unbiased estimating equation without a bias correction term. 
Therefore, the following corollary can be obtained from \cite[Theorem 3]{mypaper5}.
\begin{assumption}\label{assp:gigtm}
    There exists a GIGT mixture distribution \eqref{eq:mixture_cbd} corresponding to the nonnegative generating function $g$, i.e., $C_{\rm GIGT}(\theta, \lambda, 0)<\infty$ and $C_{\rm GIGT}(\theta, \lambda, -1)<\infty$.
\end{assumption}
\begin{corollary}
    Under Assumption \ref{assp:gigtm}, the estimating equation without a bias correction term equivalently, \eqref{eq:bias0} holds if and only if
    \begin{align}
        \int_0^\infty g(t)f'(t)dt < \infty \label{eq:cbd_cond}
    \end{align}
    holds for the combination of the function $f$ and the statistical model \eqref{eq:mixture_cbd}.
\end{corollary}

\subsection{Discussion}
In this section, we showed that when using inverse divergence for estimation, two types of statistical models satisfy the unbiased estimating equation: the IGT \eqref{eq:igt} and GIGT mixture \eqref{eq:mixture_cbd} distributions, each with different conditions for the function $f$. 
Comparing conditions \eqref{eq:bregman_ig_cond} and \eqref{eq:cbd_cond}, clearly, the factor $1/\sqrt{t+1}$ applied to the integrand differs. 
This implies that with a fixed generating function $g$, the functions $f$ available for estimation differ.
When fixing $f'(t)=1$, the IGT distribution always has an expected value independent of the generating function $g$. 
In contrast, the existence of the expected value in the GIGT mixture distribution depends on the generating function $g$.
More generally, the same can be concluded  when comparing the IGT distribution with the continuous Bregman distributions, which generalize the GIGT mixture distribution.
For distributions such as continuous Bregman, elliptical, and IS distributions that are generalized by the generating function $g$, the existence of the expected value depends on the generating function $g$. 
For example, the Cauchy distribution, a particular case of elliptical distributions, does not have an expected value. 
This highlights that the IGT distribution is somewhat unique within the family of distributions generalized by the generating function $g$.

%% file: 4_Multi-dimensional_extension.tex
\section{Extension to multi-dimensional case}
In this section, we extend the problem to cases of multi-dimensional data points.
In the following, to simplify the notation, the $j$-th dimension variable is represented by a subscript $j$.
Let the Bregman divergence be given as a linear sum over the dimensions of the inverse divergence:
\begin{align}
    d_{\rm MInv}^{\bm{\lambda}}(\bm{x}, \bm{\theta}) \triangleq \sum_{j=1}^d  d_{\rm Inv}^{\lambda_j}(x_j, \theta_j) .\label{eq:multivariate_ig_div}
\end{align}
The corresponding strictly convex function is given by $\phi(\bm{x})=\sum_{j=1}^d \lambda_j/x_j$.
\begin{definition}[Multivariate IGT (MIGT) distribution]
    For $\bm{x}\in\mathbb{R}_{++}^d$, the parameters $\bm{\theta}\in\bm{\Theta}=\mathbb{R}_{++}^d$, $\bm{\lambda}\in\mathbb{R}_{++}^d$, and a nonnegative generating function $g:\mathbb{R}_+\to\mathbb{R}_+$, the MIGT distribution is defined as follows if the normalization constant $C_{\rm MIGT}$ exists:
    \begin{align}
        p(\bm{x}|\bm{\theta}, \bm{\lambda}) &= \frac{1}{C_{\rm MIGT}}\prod_{j=1}^d \left[\sqrt{\frac{\lambda_j}{x_j^3}}\right]g(d_{\rm MInv}^{\bm{\lambda}}(\bm{x}, \bm{\theta})), \label{eq:multi_igt} \\
        C_{\rm MIGT} &= \frac{\pi^\frac{d}{2}}{\Gamma(\frac{d}{2})}\int_0^\infty g(t)t^{\frac{d-2}{2}}dt, \label{eq:mig_const}
    \end{align}
    where $\Gamma(\cdot)$ represents the gamma function.
\end{definition}
The MIGT distribution has the same normalization constant $C_{\rm MIGT}$ as elliptical distributions \cite{elliptical_book} and extends the IGT distribution to multiple dimensions. 
Setting the generating function as $g(t)=\exp(-t/2)$,  we obtain the simultaneous distribution of $d$ independent inverse Gaussian distributions.
As shown in Corollary \ref{corollary:multi-igt_expect}, the expected value of the MIGT distribution exists independent of the generating function $g$, and $\mathbb{E}[\bm{X}]=\bm{\theta}$ holds.
\begin{assumption}\label{assp:migt}
    There exists a MIGT distribution \eqref{eq:multi_igt} corresponding to the nonnegative generating function $g$, i.e., $C_{\rm MIGT}<\infty$.
\end{assumption}
\begin{theorem} \label{theorem:multi_igt}
    Let us assume $d\geq 2$.
    Under Assumption \ref{assp:migt}, the estimating equation without a bias correction term equivalently, \eqref{eq:bias0} holds if and only if 
    \begin{align}
        \int_0^\infty\int_0^\infty g(t+s)f'(t+s)\frac{t^\frac{d-3}{2}}{\sqrt{s+1}}dtds<\infty  \label{eq:multi_bregman_ig_cond}
    \end{align}
    holds for the combination of the function $f$ and the statistical model \eqref{eq:multi_igt}.
\end{theorem}
The proof of Theorem \ref{theorem:multi_igt} is provided in Appendix \ref{sec:appendix3}.
\begin{corollary}\label{corollary:multi-igt_expect}
    The expected value of a MIGT distribution \eqref{eq:multi_igt} exists, independent of the generating function $g$, and satisfies $\mathbb{E}[\bm{X}]=\bm{\theta}$.
\end{corollary}
\begin{proof}
    From Theorem \ref{theorem:multi_igt}, by setting $f'(t)=1$, we immediately obtain the following relationship:
    \begin{align}
        \int_{\mathbb{R}_+^2} g(t+s)\frac{t^{\frac{d-3}{2}}}{\sqrt{s+1}}dtds<\infty \iff \mathbb{E}_{p(\bm{x}|\bm{\theta}, \bm{\lambda})}[\bm{X}]=\bm{\theta}<\infty. \label{eq:migt_cond}
    \end{align}
    We consider the upper bound of left condition of \eqref{eq:migt_cond}:
    \begin{align}
        \int_{\mathbb{R}_+^2} g(t+s)\frac{t^\frac{d-3}{2}}{\sqrt{s+1}}dtds
         &< \int_{\mathbb{R}_+^2} g(t+s)\frac{t^\frac{d-3}{2}}{\sqrt{s}}dtds \nonumber \\
        &=\frac{\sqrt{\pi}\Gamma(\frac{d-1}{2})}{\Gamma(\frac{d}{2})}\int_0^\infty g(t)t^{\frac{d-2}{2}}dt \nonumber \\
        &\leq \frac{\pi^\frac{d}{2}}{\Gamma(\frac{d}{2})}\int_0^\infty g(t)t^{\frac{d-2}{2}}dt. \label{eq:might_upper}
    \end{align}
    The first inequality is derived from a simple comparison of the integrand. 
    The second line is obtained by setting  $m=2$, $\alpha_1=1/2$, $\alpha_2=(d-1)/2$, $u(t)=g(t)$ and applying Lemma \ref{lemma:lemma1} in Appendix \ref{sec:appendix1} to the first line. 
    Equation \eqref{eq:might_upper} is the normalization constant \eqref{eq:mig_const} of the MIGT distribution.
    Therefore, the expected value of the MIGT distribution exists independent of the generating function $g$, and $\mathbb{E}[\bm{X}]=\bm{\theta}$ holds.
    \qed
\end{proof}

%% file: 5_Conclusion.tex
\section{Conclusion}
In this paper, we discussed the conditions under which the unbiased estimating equation holds without the bias correction term for loss functions composed of a monotonically increasing function $f$ and Bregman divergence. 
In the case of inverse divergence, this scenario was satisfied by the IGT and GIGT mixture distributions (the latter is a particular case of the continuous Bregman distribution), each with different conditions of the function $f$.
In estimating the IGT distribution based on robust divergence, the bias correction terms, which are analytically intractable integrals, do not vanish. 
By defining the Bregman divergence as a linear sum over the dimensions of inverse divergence, we extended the discussion to multi-dimensional cases. 
The corresponding statistical model in this case (the MIGT distribution) was newly defined in this paper. 
Similar to the IGT distribution, we clarified that if a probability density function exists, an expected value, satisfying $\mathbb{E}[\bm{X}]=\bm{\theta}$  exists as well. 
Furthermore, we showed that the conditions for the function $f$ that can be used for estimation are provided by a double integral.

When the statistical model is the regular exponential family, the Bregman divergences ensuring the unbiased estimating equation are the squared and IS distances reported in previous studies and the inverse divergence reported in this study,  which are special cases of $\beta$-divergence \cite{Cichocki2010}. 
Although this paper did not delve into the details of robust estimation, by combining discussions from existing research, it was proven that latent biases can be reduced to approximately zero, even in the presence of outliers, and that the consistency of the estimator can be established \cite{mypaper5}.

%% file: Appendix.tex
\appendices
\section{Detailed proof of Theorem \ref{theorem:igt}}\label{sec:appendix2}
Here, we set
\begin{align}
t = d_{\rm Inv}^\lambda(x, \theta) \label{eq:t_dig}
\end{align}
and perform substitution integration, which requires that the integration interval for $x$ be divided.
The integration intervals for $x$ from $0$ to $\theta$ and from $\theta$ to $\infty$ are transformed into the integration intervals for $t$ from $\infty$ to $0$ and from $0$ to $\infty$, respectively.
The factor of the integral transform is given by
\begin{align}
        dx = \frac{\lambda^{-1}\theta^2x^2}{(x-\theta)(x+\theta)}dt. \label{eq:jacobian}
\end{align}
In addition, expanding \eqref{eq:t_dig} with respect to $x$, we obtain
\begin{align}
    \lambda x^2-\theta(\theta t + 2\lambda)x+\lambda\theta^2=0 .\nonumber
\end{align}
From the quadratic formula, we obtain
\begin{align}
    x = \frac{\theta}{2\lambda}\left((\theta t+2\lambda)\pm \sqrt{\theta t(\theta t + 4\lambda)}\right). \nonumber
\end{align}
For any $\theta$, let us denote the range of $x_{<\theta}$ as $(0,\theta)$ and the range of $x_{\geq \theta}$ as $[\theta, \infty)$.
Then, $x_{<\theta}$ and $x_{\geq \theta}$ can be expressed as the inverse function of $t$ in \eqref{eq:t_dig} as follows:
\begin{align}
    x_{<\theta}(t) = \frac{\theta}{2\lambda}\left((\theta t+2\lambda)- \sqrt{\theta t(\theta t + 4\lambda)}\right), \label{eq:x_geq_theta}\\
    x_{\geq\theta}(t) = \frac{\theta}{2\lambda}\left((\theta t+2\lambda)+ \sqrt{\theta t(\theta t + 4\lambda)}\right) .\label{eq:x_leq_theta}
\end{align}
Although these are functions of $t$, in the following discussion, $t$ may occasionally be omitted for simplicity.
Let us define the following function:
\begin{align}
    h(x) = \frac{\sqrt{x}}{x+\theta}. \nonumber
\end{align}
Substituting \eqref{eq:x_geq_theta} and \eqref{eq:x_leq_theta} into this function, we obtain
\begin{align}
    h(x_{<\theta}(t)) = h(x_{\geq\theta}(t)) = \frac{\sqrt{\lambda}}{\theta}\frac{1}{\sqrt{t+4\frac{\lambda}{\theta}}} .\label{eq:func_h}
\end{align}
Using the previous equations \eqref{eq:t_dig}, \eqref{eq:jacobian}, and \eqref{eq:func_h}, the substitution integral of \eqref{eq:proof_ig1} is calculated as
\begin{align}
    & \int_\infty^0 \frac{1}{\sqrt{x_{<\theta}^3}} g(t)f'(t)(x_{<\theta}-\theta)\frac{\lambda^{-1}\theta^2x_{<\theta}^2}{(x_{<\theta}-\theta)(x_{<\theta}+\theta)}dt \nonumber \\
    &+\int_0^\infty \frac{1}{\sqrt{x_{\geq\theta}^3}} g(t)f'(t)(x_{\geq\theta}-\theta)\frac{\lambda^{-1}\theta^2x_{\geq\theta}^2}{(x_{\geq\theta}-\theta)(x_{\geq\theta}+\theta)}dt \nonumber \\
    =&\frac{\theta^2}{\lambda} \int_0^\infty g(t)f'(t) \left[h(x_{\geq\theta}(t))-h(x_{<\theta}(t)) \right]dt \nonumber \\
    =&\frac{\theta}{\sqrt{\lambda}} \int_0^\infty g(t)f'(t) \left[\frac{1}{\sqrt{t+4\frac{\lambda}{\theta}}}-\frac{1}{\sqrt{t+4\frac{\lambda}{\theta}}}\right]dt = 0. \nonumber
\end{align}

\section{Proof of Theorem \ref{theorem:multi_igt}} \label{sec:appendix3}
    Substituting the multivariate inverse divergence \eqref{eq:multivariate_ig_div} and MIGT distribution \eqref{eq:multi_igt} into the left-hand side of \eqref{eq:bias0}, we have
    \small
        \begin{align}
            &\mathbb{E}_{p(\bm{x}|\bm{\theta}, \bm{\lambda})}\left[f'(d_{\rm MInv}^{\bm{\lambda}}(\bm{X}, \bm{\theta}))(\bm{X}-\bm{\theta})\right] \nonumber \\
            \propto&\int_{\mathbb{R}_+^d} \prod_{j=1}^d \left[\frac{1}{\sqrt{x_j^3}}\right]g(d_{\rm MInv}^{\bm{\lambda}}(\bm{x}, \bm{\theta}))f'(d_{\rm MInv}^{\bm{\lambda}}(\bm{x}, \bm{\theta}))(\bm{x}-\bm{\theta})d\bm{x}. \nonumber
        \end{align}
    \normalsize
    In the following, we set $\bar{g}(t)=g(t)f'(t)$ and focus our discussion on the $k$-th dimension:
    \small
    \begin{align}
        &\int_{\mathbb{R}_+^d} \prod_{j=1}^d \left[\frac{1}{\sqrt{x_j^3}}\right] \bar{g}(d_{\rm MInv}^{\bm{\lambda}}(\bm{x}, \bm{\theta}))(x_k-\theta_k)d\bm{x} \nonumber \\
        =&\int_{\mathbb{R}_+^{d-1}}\int_0^{\theta_k} \frac{1}{\sqrt{x_k^3}} \bar{g}(d_{\rm MInv}^{\bm{\lambda}}(\bm{x}, \bm{\theta}))(x_k-\theta_k)dx_k \prod_{j\neq k}^d \frac{1}{\sqrt{x_j^3}}dx_j  \nonumber \\
        &+\int_{\mathbb{R}_+^{d-1}}\int_{\theta_k}^\infty \frac{1}{\sqrt{x_k^3}} \bar{g}(d_{\rm MInv}^{\bm{\lambda}}(\bm{x}, \bm{\theta}))(x_k-\theta_k)dx_k \prod_{j\neq k}^d \frac{1}{\sqrt{x_j^3}}dx_j \nonumber \\
        =&\frac{\theta_k}{\sqrt{\lambda_k}}\int_{\mathbb{R}_+^{d}} \bar{g}\left(t_k+\sum_{j\neq k}^{d}  d_{\rm Inv}^{\lambda_j}(x_j, \theta_j)\right) \nonumber \\
        &\cdot\left[\frac{1}{\sqrt{t_k+4\frac{\lambda_k}{\theta_k}}}-\frac{1}{\sqrt{t_k+4\frac{\lambda_k}{\theta_k}}}\right]dt_k \prod_{j\neq k}^d \frac{1}{\sqrt{x_j^3}}dx_j = 0.\label{eq:var1}
    \end{align}
    \normalsize
    The transformation from the first equality to the second involves setting $t_k= d_{\rm Inv}^{\lambda_k}(x_k, \theta_k)$ and applying substitution integration, similar to the proof in Theorem \ref{theorem:igt}.
    For \eqref{eq:var1} to be zero, the following expression must be bounded:
    \small
    \begin{align}
        \int_{\mathbb{R}_+^{d}} \bar{g}\left(t_k+\sum_{j\neq k}^{d} d_{\rm Inv}^{\lambda_j}(x_j, \theta_j)\right)\frac{1}{\sqrt{t_k+4\frac{\lambda_k}{\theta_k}}}dt_k \prod_{j\neq k}^d \frac{1}{\sqrt{x_j^3}}dx_j . \nonumber
    \end{align}
    \normalsize
    Iteratively integrating the above expression, we obtain
    \small
        \begin{align}
            \prod_{j\neq k}^d\left[\frac{1}{\sqrt{\lambda_j}}\right]\int_{\mathbb{R}_+^d}\bar{g}\left(\sum_{j=1}^d t_j\right)\frac{1}{\sqrt{t_k+4\frac{\lambda_k}{\theta_k}}}dt_k \prod_{j\neq k}^d \frac{1}{\sqrt{t_j}}dt_j \nonumber \\
            =\prod_{j\neq k}^d\left[\frac{1}{\sqrt{\lambda_j}}\right]\frac{\pi^{\frac{d-1}{2}}}{\Gamma(\frac{d-1}{2})}\int_0^\infty \int_0^\infty \bar{g}(t+t_k)\frac{t^{\frac{d-3}{2}}}{\sqrt{t_k+4\frac{\lambda_k}{\theta_k}}}dtdt_k . \nonumber
        \end{align}
    \normalsize
    The transformation of the equation involves setting $m=d-1$ and $u(t)=\bar{g}(t+t_k)$, and using Lemma \ref{lemma:lemma2} in Appendix \ref{sec:appendix1}.
    For the same reasons as in the proof of Theorem \ref{theorem:igt}, the parameters $\lambda_k$ and $\theta_k$ can be eliminated from the integrand.
    Note that although the proof focuses on the $k$-th dimension, the bounded condition of the integral that must to be satisfied in each dimension is the same and independent of the dimension.

    Therefore, if \eqref{eq:multi_bregman_ig_cond} exists, the unbiased estimating equation \eqref{eq:bias0} holds without a bias correction term.
    Conversely, the above discussion also shows that $\mathbb{E}_{p(\bm{x}|\bm{\theta}, \bm{\lambda})}\left[\left|f'(d_{\rm MInv}^{\bm{\lambda}}(\bm{X},\bm{\theta}))(X_k-\theta_k)\right|\right]\propto 2\times\eqref{eq:multi_bregman_ig_cond}$.
    In other words, \eqref{eq:multi_bregman_ig_cond} is also a necessary condition.
    \qed

\section{Lemmas}\label{sec:appendix1}
\begin{lemma}\label{lemma:lemma1}
Given a positive integer $m\geq1$ and $\alpha_j>0$, where $(j=1,\cdots,m)$, for a nonnegative function $u$, the following relation \cite[pp.21--23]{elliptical_book} holds:
\begin{align}
    &\int_{\mathbb{R}_+^{m}} u\left(\sum_{j=1}^{m} t_j\right) \prod_{j=1}^{m} t_j^{\alpha_j-1}dt_j \nonumber \\
    &=\frac{\prod_{j=1}^{m}\Gamma(\alpha_j)}{\Gamma(\sum_{j=1}^{m}\alpha_j)} \int_0^\infty u(t) t^{\sum_{j=1}^{m}\alpha_j-1}dt. \nonumber
\end{align}
\end{lemma}

In Lemma \ref{lemma:lemma1}, by setting $\alpha_j=1/2, (j=1,\cdots,m)$, we obtain the following lemma.
\begin{lemma}\label{lemma:lemma2}
For a positive integer $m\geq1$ and a nonnegative function $u$, the following relation holds:
\begin{align}
    &\int_{\mathbb{R}_+^{m}} u\left(\sum_{j=1}^{m} t_j\right) \prod_{j=1}^{m} \frac{1}{\sqrt{t_j}}dt_j \nonumber \\
    &=\frac{\pi^{\frac{m}{2}}}{\Gamma(\frac{m}{2})} \int_0^\infty u(t) t^{\frac{m}{2}-1}dt. \nonumber
\end{align}
\end{lemma}